\documentclass[fleqn]{article}
\usepackage{amsmath}
\usepackage{amsthm}
\usepackage{amssymb}
\usepackage{enumerate}
\usepackage[all]{xy}

\theoremstyle{remark}

\theoremstyle{definition}

\theoremstyle{plain}
\newtheorem{lem}{Lemma}

\newtheorem{prop}{Proposition}
\newtheorem{conj}{Conjecture}

\DeclareMathOperator{\Tr}{Tr}

\DeclareMathOperator{\End}{End}
\DeclareMathOperator{\id}{id}
\newcommand{\ilim}{\mathop{\varprojlim}\limits}

\usepackage[pdftex]{hyperref}

\begin{document}
\begin{center}
\Large
{\bf On the algebra of local unitary invariants of pure and mixed quantum states}
\end{center}
\vspace*{-.1cm}
\begin{center}
P\'eter Vrana
\end{center}
\vspace*{-.4cm} \normalsize
\begin{center}
Department of Theoretical Physics, Institute of Physics, Budapest University of\\ Technology and Economics, H-1111 Budapest, Hungary

\vspace*{.2cm}
(\today)
\end{center}

\begin{abstract}
We study the structure of the inverse limit of the graded algebras of local unitary
invariant polynomials using its Hilbert series. For $k$ subsystems, we conjecture
that the inverse limit is a free algebra and the number of algebraically independent
generators with homogenous degree $2m$ equals the number of conjugacy classes of
index $m$ subgroups in a free group on $k-1$ generators.

Similarly, we conjecture that the inverse limit in the case of $k$-partite mixed
state invariants is free and the number of algebraically independent
generators with homogenous degree $m$ equals the number of conjugacy classes of
index $m$ subgroups in a free group on $k$ generators. The two conjectures are
shown to be equivalent.

To illustrate the equivalence, using the representation theory of the unitary groups,
we obtain all invariants in the $m=2$ graded parts and express them in a simple form
both in the case of mixed and pure states. The transformation between the two forms
is also derived. Analogous invariants of higher degree are also introduced.
\end{abstract}

\section{Introduction}

Entanglement of a quantum state in a composite system is captured by the values
of entanglement measures, functions which are invariant under the action of some
group modelling manipulations which are performed locally, without interaction
between the subsystems. In most cases, this group is either the LU (local unitary)
group which describes local transformations that can be applied with probability
one, or the SLOCC group (stochastic local operations and classical communication)
corresponding to transformations that can be done with nonzero probability\cite{Dur}.
For a review on quantum entanglement, see \cite{Horodecki}

Previous works focus mostly on composite systems with distinguishable constituents \cite{Verstraete, Brylinski, LT, LTT},
all of which has the same Hilbert space dimension. Considering the LU group has the
advantage that invariant polynomials can be written in a form that is independent
of the Hilbert space dimension, at least when it is large enough.

The outline of the paper is as follows. In section~\ref{sec:poly} we summarize
some facts about how the algebra of polynomials in the coefficients of a state
and their conjugates decomposes to irreducible representations under the action
of the local unitary group. In section~\ref{sec:diminv} we present an exact
formula for the dimension of the subspace of degree $2m$ LU-invariant polynomials
for any number of subsystems provided the single particle state spaces are all
at least $m$ dimensional. We interpret these dimensions as that of the degree $m$
homogenous subspace of the inverse limit of the algebras of LU-invariant polynomials.
We conclude this section with a conjecture regarding the structure of the inverse
limit.

In section~\ref{sec:mixent} we refine the derivation of the stable dimension
formula, and conclude that in the case when only one single particle state space
is large enough, a relatively simple formula can also be obtained, and has
an interpretation in terms of mixed state polynomial invariants. In section~\ref{sec:mixpoly}
we formulate and prove this correspondence rigorously, and we also repeat the
inverse limit construction in the case of mixed state invariants.

In section~\ref{sec:fourth} the space of fourth order LU-invariants of multipartite
quantum systems with arbitrary dimensional single particle states is described
completely. It is shown in particular, that the dimension of this space is $2^{k-1}$
where $k$ denotes the number of subsystems. The invariants are expressed in terms
of both pure and mixed states, and the transformation relating the two is also
found. In section~\ref{sec:higher} analogous invariants with
higher degree are introduced. 

\section{Polynomial invariants under the local unitary group}\label{sec:poly}

Let $k\in\mathbb{N}$ and $n=(n_1,\ldots,n_k)\in\mathbb{N}^k$, and consider the
complex Hilbert space $\mathcal{H}_n=\mathbb{C}^{n_1}\otimes\cdots\otimes\mathbb{C}^{n_k}$
describing the pure states of a composite system with $k$ distinguishable
subsystems. The group of local unitary transformations, $LU_n=U(n_1,\mathbb{C})\times\cdots\times U(n_k,\mathbb{C})$,
acts on $\mathcal{H}$ in the obvious way, i.e. regarding $\mathbb{C}^{n_i}$ as
the standard representation of $U(n_i,\mathbb{C})$.

The functions $\mathcal{H}_n\to\mathbb{C}$ which are polynomial in the coefficients and
their conjugates with respect to an arbitrary fixed basis are in one-to-one correspondence
with the vectors in $S(\mathcal{H}_n\oplus\mathcal{H}_n^{*})$, the symmetric algebra on
$\mathcal{H}_n\oplus\mathcal{H}_n^{*}$ on which an action of $LU_n$ is induced. We are looking
for invariant functions, i.e. vectors in the symmetric algebra which are fixed under
this action.

As each graded part is fixed by $LU_n$, it suffices to look at the homogenous subspaces.
By the isomorphism
\begin{equation}
S^p(\mathcal{H}_n\oplus\mathcal{H}_n^{*})\simeq\bigoplus_{m=0}^{p}S^m(\mathcal{H}_n)\otimes S^{p-m}(\mathcal{H}_n^{*})
\end{equation}
these split further into $LU$-invariant subspaces. The action of $Z(U(n_1,\mathbb{C}))\times\{1\}\times\cdots\times\{1\}\simeq U(1)\ni \lambda$
on $S^m(\mathcal{H}_n)\otimes S^{p-m}(\mathcal{H}_n^{*})$ is multiplication by $\lambda^{m-(p-m)}=\lambda^{2m-p}$,
implying that the $m=\frac{p}{2}$ term is the only one we need to look at.
In particular, the order of an invariant polynomial is always even. Note that
this fact makes it convenient for us to use a grading on the algebra of
invariants which is different from the usual one, and take homogenous degree $m$
elements to be polynomials with degree $m$ both in the components and their
conjugates, which is actually a degree $2m$ (real) polynomial.

From now on, $m$ will denote $\frac{p}{2}$. Observe that $S^{m}(\mathcal{H}_n^{*})=S^{m}(\mathcal{H}_n)^{*}$.
Let
\begin{equation}
S^m(\mathcal{H}_n)\simeq\bigoplus_{\alpha\in A}c_{\alpha}V_{\alpha}
\end{equation}
be the decomposition to the orthogonal sum of isotypic subspaces, where $c_\alpha$ are
nonnegative integers and $A$ is a set labelling the isomorphism-classes of irreducible
representations of $LU_n$. Then
\begin{equation}
S^m(\mathcal{H}_n)\otimes S^m(\mathcal{H}_n)^{*}\simeq\bigoplus_{\alpha,\alpha'\in A}c_{\alpha}c_{\alpha'}V_{\alpha}\otimes V_{\alpha'}^{*}
\end{equation}
The multiplicity of the trivial representation in $V_{\alpha}\otimes V_{\alpha'}^{*}$ is
$1$ if $\alpha=\alpha'$ and $0$ if $\alpha\neq\alpha'$. We can conclude that for $\alpha\in A$
there are three possibilities:
\begin{enumerate}[1.]
\item $c_{\alpha}=0$, in this case we do not get any invariants
\item $c_{\alpha}=1$, in this case there is a one dimensional subspace of invariant polynomials in $V_{\alpha}\otimes V_{\alpha}^{*}\le S^m(\mathcal{H}_n)\otimes S^{m}(\mathcal{H}_n^{*})$, and we can choose a "canonical" element spanning it
\item $c_{\alpha}>1$, in this case we get a $c_{\alpha}^2$ dimensional space of invariants in which we can not find a "distinguished" basis in any obvious way
\end{enumerate}

In sections \ref{sec:fourth} and \ref{sec:higher} we will concentrate on those components for which the coefficient
$c_{\alpha}$ is $1$. The above-mentioned distinguished element can be obtained as
follows. Up to normalization, there exists a unique inner product on $S^m(\mathcal{H}_n)$
which is invariant under the induced action of the full unitary group acting on $\mathcal{H}_n$.
This follows from the fact that $S^m(\mathcal{H}_n)$ carries an irreducible representation
of this group. We choose the normalization so that for any $\psi\in\mathcal{H}_n$, the
equation $\|\psi^{m}\|=\|\psi\|^{m}$ holds, where $\psi^{m}=\psi\otimes\cdots\otimes\psi\in S^m(\mathcal{H}_n)$.
Now let $P_\alpha$ denote the orthogonal projection onto the irreducible subrepresentation
indexed by $\alpha$. The value of the distinguished invariant polynomial on $\psi\in\mathcal{H}_n$
is then $\langle\psi^{m},P_{\alpha}\psi^{m}\rangle$.

\section{Stabilized dimensions of the spaces of invariant polynomials}\label{sec:diminv}

As it was mentioned in section \ref{sec:poly}, the dimension of the space of
degree $2m$ polynomial invariants on $\mathcal{H}_n$ is
\begin{equation}
\sum_{\alpha\in A}c_\alpha^2
\end{equation}
where $A$ is an index set labelling the equivalence classes of irreducible
representations of the LU group, and $\{c_\alpha\}_{\alpha\in A}$ are the
multiplicities of the irreducible representations in $S^m(\mathcal{H}_n)$.

Let $\nu$ be a partition of $m$ (this fact will be denoted by $\nu\vdash m$).
Denoting the corresponding Schur functor by $\mathbb{S}_{\nu}$, we have the
following isomorphism \cite{FH}:
\begin{equation}\label{eq:decomp}
\mathbb{S}_{\nu}\mathcal{H}_n\simeq\bigoplus_{\lambda_1,\ldots,\lambda_k\vdash m}C_{\nu\lambda_1\ldots\lambda_k}\mathbb{S}_{\lambda_1}\mathbb{C}^{n_1}\otimes\cdots\otimes\mathbb{S}_{\lambda_k}\mathbb{C}^{n_k}
\end{equation}
where $C_{\nu\lambda_1\ldots\lambda_k}$ is the multiplicity of $V_{\nu}$ in
$V_{\lambda_1}\otimes\cdots\otimes V_{\lambda_k}$ as a representation of
the symmetric group $S_m$, where $V_{\nu}$ denotes the irreducible
$S_m$-representation indexed by the partition $\nu$. Denoting the character
of $V_\lambda$ by $\chi_\lambda$ we can write
\begin{equation}
C_{\nu\lambda_1\ldots\lambda_k}=(\chi_{\nu},\chi_{\lambda_1}\chi_{\lambda_2}\cdots\chi_{\lambda_k})_{S_m}
\end{equation}
where we denoted by $(\cdot,\cdot)_{S_m}$ the usual inner product on the space
of class functions on $S_m$. Our convention is that this inner product is
semilinear in the first and linear in the second argument, but this does not
have much effect as all the appearing characters are real.

Note that $\mathbb{S}_{\lambda_i}(\mathcal{H}_i)$ is the zero vector space if
and only if $|\lambda_i|>\dim\mathcal{H}_i$. Therefore when some of the Hilbert
spaces are less than $m$ dimensional, some of the terms in eq. (\ref{eq:decomp}) may
vanish, decreasing the dimension of the space of invariants. On the other hand,
if $\forall i:n_i\ge m$, this dimension is independent of the exact values of
the $n_i$, meaning that for sufficiently large single particle state spaces
the dimension stabilizes, and depends only on $k$ and $m$.

Let $I_{k,n}$ denote the graded algebra of polynomials over $\mathcal{H}_n$
which are invariant under the action of $LU_n$. Suppose that $n,n'\in\mathbb{N}^k$
such that $n\le n'$ with respect to the componentwise (product) order. Then we have
the inclusion $\iota_{n,n'}:\mathcal{H}_n\hookrightarrow\mathcal{H}_{n'}$ which
is the tensor product of the usual inclusions $\mathbb{C}^{n_i}\hookrightarrow\mathbb{C}^{n'_i}$
sending an $n_i$-tuple to the first $n_i$ components. We can similarly regard
$LU_n$ as a subgroup of $LU_{n'}$ which stabilizes the image of $\iota_{n,n'}$,
and thus $\iota_{n,n'}$ is an $LU_n$-equivariant linear map. Therefore it induces
a morphism of graded algebras $\varrho_{n,n'}:I_{k,n'}\to I_{k,n}$ (note that
the algebra of polynomials \emph{on} a vector space is the symmetric algebra of
its \emph{dual} space, which is a contravariant construction).

Clearly, $\iota_{n,n}$ is the identity and if $n\le n'\le n''$ then
$\iota_{n',n''}\circ\iota_{n,n'}=\iota_{n,n''}$, which implies that
$\varrho_{n,n}=\id_{I_{k,n}}$ and $\varrho_{n,n'}\circ\varrho_{n',n''}=\varrho_{n,n''}$.
The central object which we study is the inverse limit of this system of graded
algebras and their morphisms:
\begin{equation}
I_k:=\ilim_{n\in\mathbb{N}^k}I_{k,n}=\left\{(f_n)_{n\in\mathbb{N}^k}\in\prod_{n\in\mathbb{N}^k}I_{k,n}\Bigg|\forall n\le n':f_n=\varrho_{n,n'}f_{n'}\right\}
\end{equation}
Note that as the product is taken in the category of graded algebras, it consists
of sequences with bounded degree. We will call $I_k$ the algebra of LU-invariants.

\begin{lem}
Suppose that $n,n'\in\mathbb{N}^k$ and $n\le n'$. Let $m\in\mathbb{N}$ such that
for all $i$ we have $m\le n_i$. Then the restriction of $\varrho_{n,n'}$ is an
isomorphism (of vector spaces) between the spaces of homogenous degree $m$ elements
of $I_{k,n'}$ and $I_{k,n}$.
\end{lem}
\begin{proof}
As it was already noted, the dimension of the two homogenous parts is equal. We
will show that $\varrho_{n,n'}$ is injective on elements of degree at most $m$.

Suppose first that for some $1\le i\le k$, $n_i=n'_i-1$ and for $j\neq i$ $n_j=n'_j$.
Let $\{e_{1,1,\ldots,1},\ldots,e_{n'_1,\ldots,n'_k}\}$ be the basis of
$\mathcal{H}_{n'}$ formed by tensor products of standard basis elements of the
$\mathbb{C}^{n_i}$. Then the algebra of real polynomials is generated by the
coordinate functions $\{e_{1,1,\ldots,1}^*,\ldots,e_{n'_1,\ldots,n'_k}^*\}$
and their conjugates $\{\overline{e_{1,1,\ldots,1}^*},\ldots,\overline{e_{n'_1,\ldots,n'_k}^*}\}$.

Let $f\in I_{k,n'}$ be a degree $m$ homogenous polynomial such that $\varrho_{n,n'}f=0$.
This means that $f$ vanishes on the image of $\iota_{n,n'}$. Denoting by $J_{a}$
the ideal generated by elements of the form $\{e_{j_1,\ldots,j_k}\}$ such
that $j_i=a$, we can reformulate this fact as $f\in J_{n_i}$. Note that $f$ is
invariant under the action of $LU_n$, and we have the subgroup $S_{n_i}\le LU_n$
which permutes the basis elements in the $i$th factor $\mathbb{C}^{n_i}$,
therefore $f$ is also contained in the ideals $J_1,\ldots,J_{n_i}$.

But the intersection of the ideals $J_1,\ldots,J_{n_1}$ is their product,
therefore $f$ is in the ideal generated by $n_1$-fold products of the
coordinate functions. As $f$ is $LU_n$-invariant, its terms must contain
the same number of conjugate coordinate functions, and therefore its
homogenous parts of degree less then $n_1$ vanish. $m\le n_1$ implies that
$f=0$.

For the general case, observe that if $n\le n'$ then $\varrho_{n,n'}$ can be
written as a composition of the maps considered above (or is the identity in
the case of $n=n'$), and hence also injective.
\end{proof}
This lemma means that every element of $I_k$ is represented in some $I_{k,n}$
(it suffices to take $n$ to be $(m,m,\ldots,m)$ with $m$ the degree of the element),
and that if $n_{min}=\min\{n_i\}$ then the factors of $I_k$ and $I_{k,n}$ by the
ideals generated by homogenous elements of degree at least $n_{min}+1$ are
isomorphic. Therefore, the algebras $I_{k,n}$ and $I_k$ are closely related,
while the latter seems considerably simpler to study. Our next aim will be to
calculate the Hilbert series of $I_k$.

Let $d_{k,m}$ denote the stabilized dimension of degree $m$ LU-invariants for a
composite system of $k$ subsystems for which a remarkable formula was
found by Hero and Willenbring \cite{HW}. We present here a slightly different derivation for later
convenience. The value of $d_{k,m}$ can be expressed as follows:
\begin{equation}\label{eq:dkmformula}
\begin{split}
d_{k,m}
  & = \sum_{\lambda_1,\ldots,\lambda_k\vdash m}C_{(m)\lambda_1,\ldots,\lambda_k}^2  \\
  & = \sum_{\lambda_1,\ldots,\lambda_k\vdash m}(\chi_{(m)},\chi_{\lambda_1}\cdots\chi_{\lambda_k})^2  \\
  & = \sum_{\lambda_1,\ldots,\lambda_k\vdash m}(\chi_{\lambda_1}\cdots\chi_{\lambda_{k-1}},\chi_{\lambda_k})(\chi_{\lambda_k},\chi_{\lambda_1}\cdots\chi_{\lambda_{k-1}})  \\
  & = \sum_{\lambda_1,\ldots,\lambda_{k-1}\vdash m}(\chi_{\lambda_1}\cdots\chi_{\lambda_{k-1}},\chi_{\lambda_1}\cdots\chi_{\lambda_{k-1}})  \\
  & = \sum_{\lambda_1,\ldots,\lambda_{k-1}\vdash m}(\chi_{(m)},\chi_{\lambda_1}^2\cdots\chi_{\lambda_{k-1}}^2)  \\
  & = (\chi_{(m)},\sum_{\lambda_1,\ldots,\lambda_{k-1}\vdash m}\chi_{\lambda_1}^2\cdots\chi_{\lambda_{k-1}}^2)  \\
  & = (\chi_{(m)},\left(\sum_{\lambda\vdash m}\chi_{\lambda}^2\right)^{k-1})  \\
\end{split}
\end{equation}

For a finite group $G$, the sum of the squares of absolute values of the characters
of inequivalent irreducible representations gives the character of the representation
on the group algebra $\mathbb{C}G$ by conjugation. The value of this character
on $g$ is the number of the elements in $G$ which commute with $g$. Therefore,
denoting by $R$ a set of representatives of the conjugacy classes in $S_m$ we have that
\begin{equation}
\begin{split}
(\chi_{(m)},\left(\sum_{\lambda\vdash m}\chi_{\lambda}^2\right)^{k-1})
  & = \frac{1}{|G|}\sum_{g\in S_m}\chi_{conj.}(g)^{k-1}  \\
  & = \frac{1}{|G|}\sum_{g\in R}|C_{S_m}(g)|\chi_{conj.}(g)^{k-1}  \\
  & = \frac{1}{|G|}\sum_{g\in R}\frac{|S_m|}{|Z_{S_m}(g)|}|Z_{S_m}(g)|^{k-1}  \\
  & = \sum_{g\in R}|Z_{S_m}(g)|^{k-2}
\end{split}
\end{equation}
where $Z_G(g)$ denotes the centralizer, and $C_G(g)$ denotes the conjugacy class of $g$.

Conjugacy classes in $S_m$ may be conveniently labelled by cycle types, which
are $m$-tuples of nonnegative integers, the $j$th integer being the number of
$j$-cycles when an (arbitrary) element of the given conjugacy class is written
as a product of disjoint cycles. Clearly, $a=(a_1,\ldots,a_m)$ describes a cycle
type if and only if it consists of nonnegative integers and
\begin{equation}
\sum_{i=1}^{m}ia_i=m
\end{equation}
holds. We will denote this fact by $a\Vdash m$.

If $g$ is an element with cycle type $a$, then the order of its centralizer is
given by the formula
\begin{equation}
|C_{S_m}(g)|=\prod_{i=1}^{m}i^{a_i}a_i!
\end{equation}
We conclude that
\begin{equation}\label{eq:stabdim}
d_{k,m}=\sum_{a\Vdash m}\left(\prod_{i=1}^{m}i^{a_i}a_i!\right)^{k-2}
\end{equation}

The Hilbert series of the algebra $I_k$ is the formal power series
\begin{equation}
\sum_{m\ge 0}d_{k,m}t^m
\end{equation}
Using eq. (\ref{eq:stabdim}) this can be rewritten as
\begin{equation}
\begin{split}
\sum_{m\ge 0}d_{k,m}t^m
  & = \sum_{m\ge 0}\sum_{a\Vdash m}\left(\prod_{i=1}^{m}i^{a_i}a_i!\right)^{k-2}t^m  \\
  & = \sum_{a_1,a_2,\ldots\ge 0}\prod_{i=1}^{m}\left(i^{a_i}a_i!\right)^{k-2}t^{ia_i}  \\
  & = \prod_{i\ge 1}\left(\sum_{a\ge 0}(i^{a}a!)^{k-2}t^{ia}\right)  \\
  & = \prod_{d\le 1}(1-t^d)^{-u_d(F_{k-1})}
\end{split}
\end{equation}
where in the last row $u_d(G)$ denotes the number of conjugacy classes of index
$d$ subgroups of a group $G$ and $F_{k-1}$ is the free group on $k-1$ generators.
This last equality can be found in \cite{Stanley}.

The formula obtained suggests the following conjecture (here we return to the usual
grading, which differs from the previously used one by a factor of two):
\begin{conj}
The algebra of $LU$-invariants $I_k$ of $k$-partite quantum systems is free,
and the number of degree $2d$ invariants in an algebraically independent generating
set equals the number of conjugacy classes of index $d$ subgroups in the free group
on $k-1$ generators.
\end{conj}

\section{Mixed state entanglement}\label{sec:mixent}

Observe that the derivation in eq. (\ref{eq:dkmformula}) remains valid when we
only assume that $\dim\mathcal{H}_{k}\ge m$ with some modification. If $\mathcal{H}=\mathcal{H}_1\otimes\cdots\otimes\mathcal{H}_k$
where $\dim\mathcal{H}_i=n_i$ and $n_k\ge m$ then
\begin{equation}
\begin{split}
d_{(n_1,\ldots,n_{k-1}),m} & := \dim(S^{2m}(\mathcal{H}\oplus\mathcal{H}^*))^{LU}  \\
 & = (\chi_{(m)},\prod_{i=1}^{k-1}\left(\sum_{\substack{\lambda\vdash m  \\  |\lambda|\le n_i}}\chi_{\lambda}^2\right))  \\
\end{split}
\end{equation}
Moreover, in this case we do not need $n_k\to\infty$ to have stabilization for
every $m$, the condition $n_k\ge n_1\cdot n_2\cdot\cdots\cdot n_{k-1}$ is sufficient.

An important special case is when $n_i=2$ for $i\le k$. In this case the above formula
reduces to
\begin{equation}
\begin{split}
d_{(2,\ldots,2),m} & = \dim(S^{2m}(\mathcal{H}\oplus\mathcal{H}^*))^{LU}  \\
 & = (\chi_{(m)},\left(\sum_{\substack{\lambda\vdash m  \\  |\lambda|\le 2}}\chi_{\lambda}^2\right)^{k-1})  \\
\end{split}
\end{equation}

One can also find a physical interpretation of this condition. Regarding
$\mathcal{H}_S:=\mathcal{H}_1\otimes\cdots\otimes\mathcal{H}_{k-1}$ as the Hilbert space
of an \emph{open} quantum system interacting with its environment with
state space $\mathcal{H}_{ENV}:=\mathcal{H}_k$, the condition above ensures that every mixed
state over $\mathcal{H}_S$ arises as the reduced state of a pure state of
$\mathcal{H}=\mathcal{H}_S\otimes\mathcal{H}_{ENV}$.

Recall that in this case given a mixed state $\varrho\in\End(\mathcal{H}_S)$ we
can find a vector $\psi\in\mathcal{H}$ (called a purification of $\varrho$) such
that $\varrho=\Tr_{ENV}\psi\psi^*$ and that $\psi$ is unique up to transformations
of the form $id_{\mathcal{H}_S}\otimes U$ where $U\in U(\mathcal{H}_{ENV})$. As
$\Tr_{ENV}:\End(\mathcal{H}_S\otimes\mathcal{H}_{ENV})\to\End(\mathcal{H}_S)$
is $GL(\mathcal{H}_S)\times U(\mathcal{H}_{ENV})$-equivariant, we have that
purification gives a bijection between $LU$-equivalence classes of mixed states
over $\mathcal{H}_S$ and pure states in $\mathcal{H}$ with the partial trace as
inverse.

Denoting the set of mixed states over a Hilbert space $\mathcal{H}$ by 
\begin{equation}
D(\mathcal{H}):=\{A\in\End(\mathcal{H})|A\ge 0,\Tr A\le 1\}
\end{equation}
and the set of unit vectors by
\begin{equation}
P(\mathcal{H}):=\{\psi\in\mathcal{H}|\|\psi\|^2=1\}
\end{equation}
we can write the commutative diagram
\begin{equation}
\xymatrix{P(\mathcal{H}) \ar@{->>}[r]^{\Tr_{ENV}\circ P}\ar@{->>}[d]  &  D(\mathcal{H}_S)\ar@{->>}[d]  \\
          P(\mathcal{H})/LU \ar@{-->}[r]^{\sim}  &  D(\mathcal{H}_S)/LU}
\end{equation}
where $P:\mathcal{H}\to\End(\mathcal{H})$ is defined by $\psi\mapsto \psi\psi^*$,
$LU=U(\mathcal{H}_1)\times\cdots\times U(\mathcal{H}_{k-1})\times U(\mathcal{H}_k)$
is the local unitary group acting in the obvious way and the vertical arrows
are the factor maps. Note that when $n_k< n_1\cdot n_2\cdot\cdots\cdot n_{k-1}$,
the lower horizontal map making this diagram commutative (as well as $\Tr_{ENV}\circ P$)
fails to be surjective.

\section{Mixed state polynomial invariants}\label{sec:mixpoly}

We have seen that the $LU$-equivalence problem of mixed states can be reduced
to the $LU$-equivalence problem of pure states. To deal with this latter problem,
one usually seeks for (real) polynomials on the Hilbert space of the composite
system which are invariant under the induced action of the $LU$-group. The reason
for this is that invariant polynomials are relatively easy to handle, while still
separate the orbits.

In the case of mixed states, polynomial invariants are not the ones which are
physically most important for the quantification of entanglement, but they are
equally well-suited for determining whether two states are equivalent as in the
case of pure states. Using the fact that $\Tr_{ENV}\circ P:P(\mathcal{H})\to D(\mathcal{H}_S)$
is an equivariant polynomial function (of degree 2), from a polynomial invariant
$f:D(\mathcal{H}_S)\to\mathbb{C}$ on pure states we can always construct one on
 mixed states, namely $f\circ\Tr_{ENV}\circ P$. Similarly, if we are given an
invariant $g:P(\mathcal{H})\to\mathbb{C}$, we can pull it back via the isomorphism
$D(\mathcal{H}_S)/LU\to P(\mathcal{H})/LU$ to obtain an invariant on mixed states:
\begin{equation}
\xymatrix{  \mathbb{C}  &  \\
          P(\mathcal{H}) \ar@{->>}[r]^{\Tr_{ENV}\circ P}\ar@{->>}[d]\ar[u]^{g}  &  D(\mathcal{H}_S)\ar@{->>}[d]\ar@{-->}[ul] \\
          P(\mathcal{H})/LU \ar[r]^{\sim}\ar@/^2pc/[uu]  &  D(\mathcal{H}_S)/LU  }
\end{equation}
The two constructions are clearly inverses of each other, but it is not clear
that $f$ is polynomial whenever $f\circ\Tr_{ENV}\circ P$ is polynomial.

To prove this, observe that the map $S^{m}(\End(\mathcal{H}_S))\to S^{2m}(\mathcal{H}\oplus\mathcal{H}^*)$
defined by $f\mapsto f\circ\Tr_{ENV}\circ P_2$ is an injective linear map
where $P=P_2\circ P_1$ and $P_1:\mathcal{H}\to\mathcal{H}\oplus\mathcal{H}^*$ is
defined by $\psi\mapsto\psi\oplus\psi^*$ while $P_2:\mathcal{H}\oplus\mathcal{H}^*\to\mathcal{H}\otimes\mathcal{H}^*$
is defined by $\psi\oplus\varphi^*\mapsto\psi\varphi^*$. As the
appearing vector spaces are by assumption finite dimensional, we need to show
that these dimensions are equal.

Let $\mathcal{H}_{S}=\mathcal{H}_{1}\otimes\cdots\otimes\mathcal{H}_{k-1}$ be the
state space of a composite quantum system, and $\mathcal{H}_{ENV}$ the state space
of its environment as before, and let $\mathcal{H}=\mathcal{H}_{S}\otimes\mathcal{H}_{ENV}$
denote the Hilbert space of the joint system composed from the two. Then
\begin{equation}
\begin{split}
S^{2m}(\mathcal{H}\oplus\mathcal{H}^{*})^{U(\mathcal{H}_{ENV})}
  & \simeq S^m(\mathcal{H})\otimes S^m(\mathcal{H}^{*})^{U(\mathcal{H}_{ENV})}  \\
  & \simeq \left(\bigoplus_{\lambda,\lambda'\vdash m}\mathbb{S}_{\lambda}\mathcal{H}_{S}\otimes\mathbb{S}_{\lambda}\mathcal{H}_{ENV}\otimes\mathbb{S}_{\lambda'}\mathcal{H}_{S}^{*}\otimes\mathbb{S}_{\lambda'}\mathcal{H}_{ENV}^*\right)^{U(\mathcal{H}_{ENV})}  \\
  & \simeq \bigoplus_{\lambda\vdash m}\left(\mathbb{S}_{\lambda}\mathcal{H}_{S}\otimes\mathbb{S}_{\lambda}\mathcal{H}_{ENV}\otimes\mathbb{S}_{\lambda}\mathcal{H}_{S}^{*}\otimes\mathbb{S}_{\lambda}\mathcal{H}_{ENV}^*\right)^{U(\mathcal{H}_{ENV})}  \\
  & \simeq \bigoplus_{\lambda\vdash m}\mathbb{S}_{\lambda}\mathcal{H}_{S}\otimes\mathbb{S}_{\lambda}\mathcal{H}_{S}^*  \\
  & \simeq S^m(\mathcal{H}_{S}\otimes\mathcal{H}_{S}^*)= S^{m}(\End(\mathcal{H}_{S}))
\end{split}
\end{equation}
as $GL(\mathcal{H}_S)\times U(\mathcal{H}_{ENV})$-modules where we have used that
\begin{equation}
\left(\mathbb{S}_\lambda\mathcal{H}_{ENV}\otimes\mathbb{S}_\lambda\mathcal{H}_{ENV}^*\right)^{U(\mathcal{H}_{ENV})}\simeq\left\{\begin{array}{ll}
\mathbb{C} & \textrm{if }\dim\mathcal{H}_{ENV}\ge|\lambda|  \\
0 & \textrm{if }\dim\mathcal{H}_{ENV}<|\lambda|
\end{array}\right.
\end{equation}
with the trivial representation on $\mathbb{C}$ and that $\mathbb{S}_\lambda\mathcal{H}_S\simeq 0$
iff $|\lambda|>\dim\mathcal{H}_S\le\dim\mathcal{H}_{ENV}$

Similarly to the case of pure state invariants, we may construct the inverse limit
of all the algebras of invariant polynomials with a fixed number of subsystems.

For $k\in\mathbb{N}$ and $n\in\mathbb{N}^k$, let $I^{mixed}_{k,n}$ denote the
algebra of $LU_n$-invariant (real) polynomials on $\End(\mathcal{H}_n)$. The
inclusions $\iota_{n,n'}:\mathcal{H}_n\to\mathcal{H}_{n'}$ induce also in this
case the maps $\varrho_{n,n'}:I^{mixed}_{k,n'}\to I^{mixed}_{k,n}$ with similar
composition properties. Let us consider the inverse limit of this system:
\begin{equation}
I^{mixed}_k:=\ilim_{n\in\mathbb{N}^k}I^{mixed}_{k,n}=\left\{(f_n)_{n\in\mathbb{N}^k}\in\prod_{n\in\mathbb{N}^k}I^{mixed}_{k,n}\Bigg|\forall n\le n':f_n=\varrho_{n,n'}f_{n'}\right\}
\end{equation}

As in the case of pure state invariants, this inverse limit also has the property
that every element is represented in some $I^{mixed}_{k,n}$, and that there is no
difference between $I^{mixed}_{k,n}$ and $I^{mixed}_{k}$ when we consider only
elements with degree at most the minimum of the dimensions $\{n_i\}_{1\le i\le k}$.

The above-shown correspondence between mixed and pure state invariants is clearly
reflected in the isomorphism $I^{mixed}_{k}\simeq I_{k+1}$ induced by the isomorphisms
\begin{equation}
I^{mixed}_{k,(n_1,\ldots,n_k)}\simeq I_{k+1,(n_1,\ldots,n_k,n_1\cdot\ldots\cdot n_k)}
\end{equation}
described above. Note that the grading of $I^{mixed}_{k}$ is the usual which is
to be contrasted with the extra factor of $2$ in the grading of $I_{k}$. With this
convention, the map $f\mapsto f\circ\Tr_{ENV}\circ P$ respects the degree.

We can also formulate our conjecture in terms of mixed state invariants:
\begin{conj}
The algebra of mixed state $LU$-invariants $I^{mixed}_k$ of $k$-partite quantum systems is free,
and the number of degree $d$ invariants in an algebraically independent generating
set equals the number of conjugacy classes of index $d$ subgroups in the free group
on $k$ generators.
\end{conj}

\section{Degree four invariants}\label{sec:fourth}

In this section we investigate the $m=2$ case. We will see that this case is
special in that $S^2{\mathcal{H}}$ is the direct sum of irreducible representations
with multiplicity at most one. Indeed, we only need to observe that the two
representations of $S_2$, the trivial and the alternating one are both one
dimensional, hence the tensor product of their arbitrary powers is irreducible.
To be specific, as $\chi_{(2)}$ is constant $1$ and $\chi_{(1,1)}$ is $1$ on
the identity and $-1$ on the other element,
\begin{equation}
(\chi_{(2)},\chi_{(2)}^a\chi_{(1,1)}^b)_{S_2}
   = \frac{1}{2}(1+(-1)^b)=\left\{\begin{array}{ll}
1 & \textrm{$b$ even}  \\
0 & \textrm{$b$ odd}
\end{array}\right.
\end{equation}

Consequently, the irreducible components of $S^2(\mathcal{H})$ can be indexed by
even-element subsets of $\{1,\ldots,k\}$, with a bijection sending the set $A$
to the component in which $\mathcal{H}_{j}$'s alternating power appears iff $j\in A$.
The subspace corresponding to $A$ will be denoted by $V_A$. It follows that the
dimension of the space of fourth order $G$-invariant polynomials is $2^{k-1}$.

Our next aim is to construct an orthonormal basis in $V_A$ for each possible subset $A$.
We would like to express elements of $S^2(\mathcal{H})$ in terms of a computational basis
in $\mathcal{H}$. 
Let $\{e_{j,i}|1\le j\le k,1\le i\le n_j\}$ be a set of vectors such that
$\{e_{j,i}\}_{1\le i\le n_j}$ is an orthonormal basis in $\mathcal{H}_j$. Let us
now introduce the following short notation: $e_{i_1 i_2 \ldots i_k}:=e_{1,i_1}\otimes e_{2,i_2}\otimes\ldots\otimes e_{k,i_k}\in\mathcal{H}$,
where $1\le i_j\le n_j$ (for all $1\le j\le k$). The set of vectors of this form is an
orthonormal basis in $\mathcal{H}$. Elements of the symmetric algebra $S(\mathcal{H})$
are polynomials in these vectors, in particular, a vector of $S^m(\mathcal{H})$ is
a degree $m$ homogenous polynomial.

Let $i_{0,1},i_{0,2},\ldots,i_{0,k},i_{1,1},i_{1,2},\ldots,i_{1,k}$ be fixed integers such that
$1\le i_{0,j}\le i_{1,j}\le n_j$ for all $1\le j\le k$, and $i_{0,j}\neq i_{1,j}$ whenever $j\in A$.
Let us now consider the vector
\begin{equation}
v=\sum_{b_1,\ldots,b_k=0}^1(-1)^{|A\cap B|}e_{i_{b_1,1} i_{b_2,2}\ldots i_{b_k,k}}e_{i_{1-b_1,1} i_{1-b_2,2}\ldots i_{1-b_k,k}}
\end{equation}
where $B=\{j\in[k]|b_j=1\}$. We claim that this is an element of $V_A$,
moreover, vectors of this type form a basis of $V_A$ and are pairwise orthogonal.

Clearly, when we construct two vectors $v$ and $v'$ this way starting from different sets
of indices, then not only $v$ and $v'$ are orthogonal, but any term appearing in the above
expression of $v$ is orthogonal to any term in $v'$. It is also easy to see that the span
of these vectors is $G$-invariant. The highest weights can be read off from the vector with
smallest possible indices, namely, for $j\notin A$, the highest weight for the $j$th factor
in $G$ is $(2)$, while for $j\in A$ it is $(1,1)$. Therefore, vectors of this type span
$V_A$. Note, that this is consistent with the fact that the number of admissible sets of
indices is
\begin{equation}
\prod_{j\in\{1,\ldots,k\}\setminus A}\binom{n_j+1}{2}\prod_{j\in A}\binom{n_j}{2}=\dim V_A
\end{equation}
for a fixed subset $A$.

We calculate next the norm squared of the
elements of this basis. The sum has $2^k$ terms, but they are not necessarily distinct.
More precisely, each term appears with the same multiplicity, which is easily seen to
be $2^{c+1}$ if $c:=\{j\in[k]|i_{0,j}=i_{1,j}\}<k$ and $2^k$ if $c=k$.
The latter case can only be realized if $A=\emptyset$. The norm of a single term is
$1$ if $c=k$, and $\frac{1}{\sqrt{2}}$ otherwise. To sum up, the norm squared of $v$ is
\begin{equation}
\|v\|^2=\left\{\begin{array}{ll}
\frac{2^k}{2^{c+1}}(2^{c+1})^{2}\frac{1}{2}=2^{k+c}  &  \textrm{if $c<k$}  \\
(2^{k})^2 &  \textrm{if $c=k$}
\end{array}\right.
\end{equation}
A formula which gives back both cases is $\|v\|^2=2^{k+c}$.

The invariant associated to the subrepresentation $V_A$ is therefore given by
\begin{equation}
I_A(\psi)=2^{-k}\sum_{\substack{1\le i^0_1\le i^1_1\le n_i  \\ \vdots \\ 1\le i^0_k\le i^1_k\le n_k}}2^{-c}\left|\sum_{b_1,\ldots,b_k=0}^1(-1)^{|A\cap B|}\psi_{i_{b_1,1} i_{b_2,2}\ldots i_{b_k,k}}\psi_{i_{1-b_1,1} i_{1-b_2,2}\ldots i_{1-b_k,k}}\right|^2
\end{equation}
with $B$ and $c$ as above, and $\psi=\sum\psi_{i_1 i_2 \ldots i_k}e_{i_1 i_2 \ldots i_k}$.

According to section \ref{sec:mixpoly}, we can also find $2^{k-1}$ linearly
independent polynomial invariants of degree $2$ on the space of mixed states
of a $k-1$-partite quantum system. A convenient choice is the following one.
Let $A\subseteq[k-1]$ and let
\begin{equation}
\Tr_{A}:\End(\mathcal{H}_{1}\otimes\cdots\otimes\mathcal{H}_{k-1})\to\End(\bigotimes_{\substack{i=1  \\  i\notin A}}^{k}\mathcal{H}_i)
\end{equation}
denote the partial trace over the subsystems whose index is in $A$. Then
for a mixed state $\varrho\in\End(\mathcal{H}_{1}\otimes\cdots\otimes\mathcal{H}_{k-1})$
\begin{equation}\label{eq:mixsec}
\varrho\mapsto\Tr((\Tr_{A}\varrho)^2)
\end{equation}
is clearly a local unitary invariant, and as $A$ runs through different subsets
of $[k-1]$, we obtain this way a linearly independent set of local unitary invariants.

To relate the two bases, let us introduce an environment to the $k-1$ particle
system, and note that $\Tr((\Tr_{A}\varrho)^2)=\Tr((\Tr_{A\cup\{k\}}\psi\psi^*)^2)=\Tr((\Tr_{[k-1]\setminus A}\psi\psi^*)^2)$
for some $\psi$.
We have the following proposition:
\begin{prop}
Let $I_A(\psi)$ be as above and $J_A(\varrho)=\Tr((\Tr_{A}\varrho)^2)$. Then
\begin{equation}
J_S(\psi\psi^*)=\sum_{A\subseteq[k]}(-1)^{|A\cap S|}I_A(\psi)
\end{equation}
and
\begin{equation}
I_A(\psi\psi^*)=2^{-k}\sum_{B\subseteq[k]}(-1)^{|A\cap B|}J_B(\psi)
\end{equation}
\end{prop}
\begin{proof}
Let $\psi=\sum_{i_1,\ldots,i_k}\psi_{i_1\ldots i_k}e_{i_1}\otimes\cdots\otimes e_{i_k}$. Then
\begin{equation}
\begin{split}
\sum_{A\subseteq[k]}(-1)^{|A\cap S|}I_A(\psi)
  & = \sum_{\substack{i_1^0\le i_1^1 \\ \vdots \\ i_k^0\le i_k^1}}2^{-(k+c)}\sum_{A\subseteq[k]}(-1)^{|A\cap S|}\sum_{\substack{b_1,\ldots,b_k \\ b'_1,\ldots,b'_k}}(-1)^{|A\cap B|}(-1)^{|A\cap B'|}\cdot  \\
 & \phantom{=}\cdot\psi_{i_1^{b_1}\ldots i_k^{b_k}}\psi_{i_1^{1-b_1}\ldots i_k^{1-b_k}}\overline{\psi_{i_1^{b'_1}\ldots i_k^{b'_k}}}\overline{\psi_{i_1^{1-b'_1}\ldots i_k^{1-b'_k}}}  \\
  & = \sum_{\substack{i_1^0\le i_1^1 \\ \vdots \\ i_k^0\le i_k^1}}2^{-c}\sum_{b_1,\ldots,b_k=0}^1 \psi_{i_1^{b_1}\ldots i_k^{b_k}}\overline{\psi_{i_1^{b_1+s_1}\ldots i_k^{b_k+s_k}}}\psi_{i_1^{b_1+s_1}\ldots i_k^{b_k+s_k}}\overline{\psi_{i_1^{b_1}\ldots i_k^{b_k}}}  \\
  & = \sum_{\substack{i_1^0,\ldots i_k^0 \\ i_1^1,\ldots,i_k^1}}2^{-k}\sum_{b_1,\ldots,b_k=0}^1 \psi_{i_1^{b_1}\ldots i_k^{b_k}}\overline{\psi_{i_1^{b_1+s_1}\ldots i_k^{b_k+s_k}}}\psi_{i_1^{b_1+s_1}\ldots i_k^{b_k+s_k}}\overline{\psi_{i_1^{b_1}\ldots i_k^{b_k}}}  \\
  & = \sum_{\substack{i_1^0,\ldots i_k^0 \\ i_1^1,\ldots,i_k^1}} \psi_{i_1^{0}\ldots i_k^{0}}\overline{\psi_{i_1^{s_1}\ldots i_k^{s_k}}}\psi_{i_1^{s_1}\ldots i_k^{s_k}}\overline{\psi_{i_1^{0}\ldots i_k^{0}}}  \\
  & = \Tr(\Tr_{S}\psi\psi^{*})^2 = J_S(\psi\psi^*)
\end{split}
\end{equation}
where $B=\{i|b_i=1\}$ and similarly for $B'$ and $S$, superscripts are understood modulo 2, and we have used the identity
\begin{equation}
\begin{split}
\sum_{A\subseteq[k]}(-1)^{|A\cap S|+|A\cap B|+|A\cap B'|}
  & = \sum_{A\subseteq[k]}\prod_{a\in A}(-1)^{1_{a\in S}+1_{a\in B}+1_{a\in B'}}  \\
  & = \prod_{a\in[k]}\left(1+(-1)^{1_{a\in S}+1_{a\in B}+1_{a\in B'}}\right)  \\
  & = \left\{\begin{array}{ll}
2^k & \textrm{if $B'=S\triangle S'$}  \\
0 & \textrm{else}
\end{array}\right.
\end{split}
\end{equation}
where $\triangle$ denotes the symmetric difference.

In the other direction:
\begin{equation}
\begin{split}
2^{-k}\sum_{B\subseteq[k]}(-1)^{|A\cap B|}J_{B}(\psi\psi^*)
  & = 2^{-k}\sum_{B\subseteq[k]}(-1)^{|A\cap B|}\sum_{S\subseteq[k]}(-1)^{|B\cap S|}I_S(\psi)  \\
  & = 2^{-k}\sum_{S\subseteq[k]}I_S(\psi)\sum_{B\subseteq[k]}(-1)^{|A\cap B|+|B\cap S|}  \\
  & = I_{A}(\psi)
\end{split}
\end{equation}
since
\begin{equation}
\begin{split}
\sum_{B\subseteq[k]}(-1)^{|A\cap B|+|B\cap S|}
  & = \sum_{B\subseteq[k]}\prod_{b\in B}(-1)^{1_{b\in A}+1_{b\in S}}  \\
  & = \prod_{b\in[k]}(1+(-1)^{1_{b\in A}+1_{b\in S}})  \\
  & = \left\{\begin{array}{ll}
2^k & \textrm{if $A=B$}  \\
0 & \textrm{else}
\end{array}\right.
\end{split}
\end{equation}
\end{proof}

Note that one can easily form an entanglement monotone from each $J_A$, as follows
\cite{PJLove}:
\begin{equation}
\eta_A=\frac{2^{|S|}}{2^{|S|}-1}(1-J_A)
\end{equation}
and this quantity takes its values between $0$ and $1$.

From the Brennen \cite{Brennen} form it also follows that the Meyer-Wallach
measure \cite{Meyer} can be expressed as
\begin{equation}
\begin{split}
Q(\psi)
  & = 2-\frac{2}{k}\sum_{i=1}^{k}J_{\{i\}}  = 2-\frac{2}{k}\sum_{i=1}^{k}\sum_{A\subseteq[k]}(-1)^{|A\cap\{i\}|}I_A  \\
  & = 2-\frac{2}{k}\sum_{A\subseteq[k]}I_A\sum_{i=1}^{k}(-1)^{|A\cap\{i\}|}   = 2-\frac{2}{k}\sum_{A\subseteq[k]}(k-2|A|)I_A  \\
  & = 2\sum_{A\subseteq[k]}I_A-\frac{2}{k}\sum_{A\subseteq[k]}(k-2|A|)I_A   = \sum_{A\subseteq[k]}\left(2-\frac{2}{k}(k-2|A|)\right)I_A  \\
  & = \sum_{A\subseteq[k]}\frac{4|A|}{k}I_A
\end{split}
\end{equation}

\section{Invariants of higher order}\label{sec:higher}

We turn to the general $m\ge 2$ case, and construct invariants analogous to the ones
introduced in section \ref{sec:fourth}. We consider the two simplest irreducible representations
of $S_m$, the trivial and the alternating one, corresponding to symmetric and alternating
powers in eq. (\ref{eq:decomp}). It is easy to see that
\begin{equation}
V_{(m)}\otimes\cdots\otimes V_{(m)}\otimes \underbrace{V_{(1,1,\ldots,1)}\otimes\cdots\otimes V_{(1,1,\ldots,1)}}_{\textrm{$b$ times}}\simeq\left\{\begin{array}{ll}
V_{(m)}  &  \textrm{if $b$ is even}  \\
V_{(1,1,\ldots,1)}  &  \textrm{if $b$ is odd}  \\
\end{array}\right.
\end{equation}
Consequently, in the decomposition of $S^m(\mathcal{H})$ to irreducibles we can find
$\mathbb{S}_{\lambda_1}\mathcal{H}_1\otimes\ldots\otimes\mathbb{S}_{\lambda_k}\mathcal{H}_k$
with multiplicity one (zero) when among the $\lambda$s only $(m)$ and $(1,1,\ldots,1)$ are
present and the latter appears an even (odd) number of times (note, that tensor product
is commutative up to isomorphism).

This means that again we have a family of local unitary invariants, labelled by
even-element subsets of $\{1,\ldots,k\}$, mapped bijectively to the set of irreducible
subrepresentations of $S^m(\mathcal{H})$ built up from symmetric and alternating powers
of the representations $\mathcal{H}_j$ with an even number of alternating powers. Again,
the subspace corresponding to the subset $A\subseteq\{1,\ldots,k\}$ will be denoted by
$V_A$.

For a fixed subset $A$, let $(i_{j,l})_{1\le j\le k,1\le l\le m}$ be integers such that
if $j\notin A$, then $1\le i_{j,1}\le i_{j,2}\le\ldots\le i_{j,m}\le n_j$ and if
$j\in A$, then $1\le i_{j,1}< i_{j,2}<\ldots< i_{j,m}\le n_j$. From these indices we
can form the vector
\begin{equation}
\sum_{\pi_1,\ldots,\pi_k\in S_m}\prod_{j=1}^{m}\chi_{\lambda_j}(\pi_j)e_{i_{1,\pi_1(1)}i_{2,\pi_2(1)}\ldots i_{k,\pi_k(1)}}e_{i_{1,\pi_1(2)}i_{2,\pi_2(2)}\ldots i_{k,\pi_k(2)}}\ldots e_{i_{1,\pi_1(m)}i_{2,\pi_2(m)}\ldots i_{k,\pi_k(m)}}
\end{equation}
where $\lambda_j=(m)$ if $j\notin A$ and $\lambda_j=(1,1,\ldots,1)$ if $j\in A$, 
$\chi_\lambda$ is the character of the corresponding Specht module, i.e. constant
$1$ if $\lambda=(m)$ and the sign of the permutation if $\lambda=(1,1,\ldots,1)$.
$v$ is then an element of $V_A$, vectors of this form span $V_A$ and they are pairwise
orthogonal for different (multi)sets of indices. As there does not seem to exist
a simple formula for the norm squared of these vectors, we cannot give the general
form of the corresponding invariant.

It would be interesting to relate these invariants with those obtained from
invariants on mixed states of a $k-1$-particle quantum system via the isomorphism
$f\mapsto f\circ\Tr_{ENV}\circ P$. It is interesting to note that formulas
analogous to the $m=2$ case stated in the proposition above but with the exponent
$m$ in eq. (\ref{eq:mixsec}) instead of $2$ do not hold.

\section{Conclusion}

Starting from the algebras of local unitary invariant polynomials on $k$-partite
quantum systems with arbitrary dimensional single particle state spaces we have
introduced a single algebra, their inverse limit (in the category of graded algebras)
with respect to the homomorphisms induced by inclusions. Unlike in the case of the
individual algebras, its Hilbert series can be given by a relatively simple formula.
Rewriting it in the form of an infinite product, we formulated the conjecture that
this algebra is free with $u_d(F_{k-1})$ homogenous generators of degree $d$.

The same program can be carried out in the case of mixed states, and the resulting
algebras turn out to be the same as in the case of pure states with a shift in
the number of subsystems. This phenomenon is already present in the case of a single
$k-1$-partite mixed state over a Hilbert space with finite dimension and its
purification, a $k$-partite pure state in the state obtained by tensoring with
a sufficiently large Hilbert space representing the environment.

As an illustration we have given all fourth order local unitary invariants of pure
states of quantum systems with distinguishable constituents having arbitrary (and
not necessarily equal) dimensional Hilbert spaces, as well as all degree two invariants
of mixed states with one less subsystems. The two vector spaces are given with
two bases which are particularly simple from the two points of view, and the linear
transformation relating the two bases is calculated.
Analogous invariants of higher order have also been constructed.

Note that the conjecture is easily verified to be true in the simplest $k=2$ case.
The algebra of LU-invariants of a bipartite quantum system are known to be
generated by the traces of the positive integer powers of the reduced density matrix,
and therefore we have one generator in every even homogenous part, which are
algebraically independent. On the other hand, the free group on one generator is
isomorphic to $\mathbb{Z}$, which clearly has exactly one index $m$ subgroup for
all $m\ge 1$.

Even if the general case turned out to be true, it would remain a great challenge
to find a way to map a conjugacy class of subgroups of $F_{k-1}$ to an LU-invariant
in an algebraically independent generating set.

It is interesting to note in the light of this equivalence of pure and mixed state
classification that a general theory of bipartite mixed entanglement is still missing,
and for pure tripartite states, also only partial results exist. It is remarkable that
the well understood case of two-qubit mixed states is related to the pure
states of a $2\times 2\times n$ system, for which a complete classification has
also been obtained \cite{Miyake}.

\end{document}